\documentclass{eptcs}
\providecommand{\event}{EXPRESS/SOS 2012}
\usepackage{amsmath,amsthm,amssymb}
\usepackage{xspace,url,hyperref}

\newcommand{\pack}[2]{\{#1\}_{#2}}
\newcommand{\DECRYPT}[3]
  {\mathsf{decrypt}\ #1\ \mathsf{as}\ #2\,  \mathsf{in}\, #3}
\newcommand{\leftSQ}{\boldsymbol{[}\kern-.20em\boldsymbol{[}\,}
\newcommand{\rightSQ}{\boldsymbol{]}\kern-.20em\boldsymbol{]}\,}
\newcommand{\encSQ}[1]{\leftSQ{#1}\,\rightSQ}
\newcommand{\encSQA}[1]{\leftSQ{#1}\,\boldsymbol{]}\kern-.20em\boldsymbol{]}_Z}

\newcommand{\ignore}[1]{}
\newcommand{\review}{\ignore}

\newcommand{\lts}[1]
{ \setbox0=\hbox{$\ {\scriptstyle#1}\ $}
  \setbox1=\hbox{$\rightarrow$}
  \loop\setbox1=\hbox{$-$\kern-0.3em\unhbox1}\ifdim\wd1<\wd0\repeat
  \hbox{$\ \mathop{\box1}\limits^{#1}\ $}
}
\newcommand{\Lts}[1]
{ \setbox0=\hbox{$\ {}^{#1}\ $}
  \setbox1=\hbox{$\Rightarrow$}
  \loop\setbox1=\hbox{$=$\kern-0.3em\unhbox1}\ifdim\wd1<\wd0\repeat
  \hbox{$\ \mathop{\box1}\limits^{#1}\ $}
}

\newcommand{\accept}[1]{:{#1}}
\newcommand{\forbids}[1]{\div{#1}}
 
\newcommand{\keyword}[1]{\textsf{\upshape\small #1}\xspace}

\newcommand{\ifk}{\keyword{if}}
\newcommand{\thenk}{\keyword{then}}
\newcommand{\elsek}{\keyword{else}} 

\newcommand{\PAR}{\mid}

\newcommand{\INACT}{\mathbf 0}  
\newcommand{\oput}[1]{\langle{#1}\rangle}
\newcommand{\SENDn}[2]{\overline{#1}\langle{#2}\rangle}
\newcommand{\SEND}[2]{\SENDn{#1}{#2}.} 
\newcommand{\RECEIVEn}[2]{{#1}({#2})}
\newcommand{\RECEIVE}[2]{\RECEIVEn{#1}{#2}.}
\newcommand{\RECEIVEnt}[2]{{#1}[{#2}]}
\newcommand{\RECEIVET}[2]{\RECEIVEnt{#1}{#2}.}
\newcommand{\IF}[2]{\ifk\,#1\,\thenk\,#2\,\elsek\,}
\newcommand{\IFs}[1]{\ifk\,#1\,\thenk\, }
\newcommand{\NR}[1]{(\keyword{new}\, #1)}
\newcommand{\SR}[1]{[\keyword{hide}\, #1]} 
\newcommand{\spy}[1]{\keyword{spy}.{#1}}
\newcommand{\spyb}[1]{\keyword{spy}\accept S.{#1}}
 
\newcommand{\osred}{\,\rightarrow\,}                
\newcommand{\subs}[2]{\{{#1}/{#2}\}} 
\newcommand{\barb}{\mathrel{\!\downarrow}}
\newcommand{\nbarb}{\mathrel{\!\not\downarrow}}
\newcommand{\Barb}{\mathrel{\!\Downarrow}}
\newcommand{\nBarb}{\!\not\Downarrow}
\newcommand{\CR}{\,{\cal R}\,}

\newcommand{\fv}{\operatorname{fn}}      
\newcommand{\bv}{\operatorname{bn}}      
 
\newcommand{\subject}{\operatorname{subj}}
\newcommand{\object}{\operatorname{obj}}

\newcommand{\rulename}[1]{[\text{\sc #1}]\xspace} 
\newcommand{\reductionrulename}[1]{\rulename{R-{#1}}}
\newcommand{\rcom}{\reductionrulename{Com}}
\newcommand{\rtcom}{\reductionrulename{T-Com}} 
 
\newcommand{\rres}{\reductionrulename{New}}
\newcommand{\rhide}{\reductionrulename{Hide}}
\newcommand{\rstruct}{\reductionrulename{Struct}} 
\newcommand{\rpar}{\reductionrulename{Par}}

\newcommand{\sreductionrulename}[1]{\rulename{RS-{#1}}} 
\newcommand{\rscom}{\sreductionrulename{Com}}
\newcommand{\rstcom}{\sreductionrulename{T-Com}} 
\newcommand{\ltsrulename}[1]{\rulename{L-{#1}}}
\newcommand{\lout}{\ltsrulename{Out}} 
\newcommand{\lopen}{\ltsrulename{Open}}
\newcommand{\lclose}{\ltsrulename{Close}} 
\newcommand{\linp}{\ltsrulename{In}}
\newcommand{\linpt}{\ltsrulename{In-T}} 
\newcommand{\lcom}{\ltsrulename{Com}} 
\newcommand{\lpar}{\ltsrulename{Par}}
\newcommand{\lrepl}{\ltsrulename{Repl}}

\newcommand{\lres}{\ltsrulename{New}} 
\newcommand{\lhide}{\ltsrulename{Hide}}
 
\newcommand{\sltsrulename}[1]{\rulename{L-{#1}}}

\newcommand{\slclose}{\sltsrulename{Close}}

\newcommand{\slspycom}{\sltsrulename{Spy-Com}}
\newcommand{\slcom}{\sltsrulename{Com}}

\newcommand{\slresn}{\sltsrulename{New}} 

\newcommand{\slhidet}{\sltsrulename{Hide}} 
\newcommand{\slspy}{\sltsrulename{Spy}}
\newcommand{\slspyn}{\sltsrulename{Spy-Res}}
 
\newcommand{\less}{\backslash} 

\newcommand{\leftB}{(\kern-.22em| }
\newcommand{\rightB}{ |\kern-.22em)}
\newcommand{\enc}[1]{\leftB{#1}\rightB}
\newcommand{\leftH}{[\kern-.35em[\, }
\newcommand{\rightH}{ \,]\kern-.35em]}
\newcommand{\encH}[1]{\leftH{#1}\rightH}

\newcommand{\grmeq}{\; ::= \;}

\newcommand{\eqdef}{\;\stackrel{\text{\scriptsize def}}{=}\;}

\newcommand{\scong}{\stackrel{\bullet}{\cong}}
\newcommand{\sapprox}{\stackrel{\bullet}{\approx}}

\newcommand{\pwd}{{\small \mathit{pwd}}} 

\newcommand{\sys}{{\small \mathit{sys}}} 

\newcommand{\shrink}[1]{}

\newcommand{\typesetproofs}[1]{} 

\title{Hide and New in the $\pi$-calculus 
\author{Marco Giunti
\institute{CITI and DI-FCT, Universidade Nova de Lisboa, Portugal\thanks{Work partially
supported by the project PTDC/EIA-CCO/117513/2010 Liveness, Statically. The work
has been done during the period that the author  spent at LIX, Ecole
Polytechnique, with the support of an ERCIM  postdoc fellowship. The author would like to
thank INRIA and ERCIM for such oppurtunity.}}  
\and Catuscia Palamidessi \qquad Frank D. Valencia\institute{INRIA Saclay and LIX, Ecole
Polytechnique, France  \thanks{Work
partially supported by the project ANR-09-BLAN-0169-01  PANDA}}}}

\def\titlerunning{Hide and New in the $\pi$-calculus }
\def\authorrunning{M.~Giunti, C.~Palamidessi and Frank~D.Valencia}  

\newtheorem{theorem}{Theorem}[section]

\newtheorem{proposition}[theorem]{Proposition}
\newtheorem{example}[theorem]{Example} 
\newtheorem{definition}[theorem]{Definition} 

\begin{document}
 
\maketitle

\begin{abstract} 
\ignore{
In the $\pi$-calculus-based formalisms for security applications, the restriction
operator
({\em new}) plays a fundamental role as it allows hiding the use of a channel within a
scope. In other words, the {\em subject of the communication} can be accessed only by the
processes in the scope of the operator. The scope, however, can be enlarged dynamically by
extruding the name. Furthermore, communication may be between remote processes. For these
reasons the actual implementation of the {\em new} operator usually relies on
non-dedicated channels that may be insecure due, for instance, to side-channel attacks. 
}
In this paper, we enrich the $\pi$-calculus with an operator for confidentiality ({\em
hide}), whose main effect is to restrict the access to the {\em object of the
communication}, thus representing confidentiality in a natural way.  The {\em hide}
operator  is meant for local communication, and it differs from {\em new} in that it
forbids the extrusion of the name and hence has a static scope. Consequently, a
communication channel in the scope of a {\em hide} can be implemented as a dedicated
channel, and it is more secure than one in the scope of a {\em new}. To emphasize the
difference, we introduce a {\em spy} context that represents a side-channel attack and
breaks some of the standard security equations for {\em new}. To formally reason on the
security guarantees provided by the \emph{hide} construct, we introduce an observational
theory
and establish stronger equivalences by relying on a proof technique based on bisimulation
semantics.

\end{abstract}

\section{Introduction}
\label{sec:intro}

The  restriction operator is present in most process calculi. Its behaviour is crucial for
\emph{expressiveness} (e.g., for specifying  unbounded linked structures, nonce generation
and locality).  
In the $\pi$-calculus~\cite{MPW92a,MPW92b}, it plays a prominent role: It provides 
for the generation and extrusion of unique names. In CCS~\cite{Milner80}, it is also
fundamental but it does not provide for name extrusion: It limits the interface of a given
process with its external world. In this paper we shall extend the $\pi$-calculus with a
hiding operator, called  {\sf hide}, that behaves similarly to the CCS restriction. The
motivation for our work comes from the realm of \emph{secrecy} and \emph{confidentiality}: we shall
argue that {\sf hide} allows us to express and guarantee secret communications.

\paragraph{Motivation.} Secrecy and confidentiality are major concerns in most  systems of communicating agents. 
Either because some of the agents are untrusted, or because the communication uses insecure channels, there 
may be the risk of sensitive information being leaked to potentially malicious entities. 
The price to pay for such security breaches may also be very high. 
It is not surprising, therefore, that secrecy and confidentiality have become  central issues in the 
formal specification and  verification of communicating systems. 

The $\pi$-calculus and especially its variants enriched with mechanisms to
express 
cryptographic operations, the  spi calculus~\cite{AbadiG99} and the applied
$\pi$-calculus~\cite{AbadiF01}, have become 
popular formalisms for security applications. They  all feature the operator {\sf new} (restriction) 
and make crucial use of it in the definition of security protocols. 
The prominent aspects of {\sf new} are the capability of creating a new channel name, 
whose use is restricted within a certain scope, and the possibility of enlarging its scope by communicating it to other processes. 
The latter property is central to the most interesting feature of the $\pi$-calculus: the \emph{mobility} of the communication structure. 

Although in principle the restriction  aspect of {\sf new}  should guarantee that the channel is used for communication 
within a secure environment only, the capability of extruding the scope leads to security problems. 
In particular, it makes it unnatural  to implement the communication using dedicated channels, and
non-dedicated channels are not  secure by default. 
The spi calculus and the applied $\pi$-calculus do not assume, indeed, any 
security guarantee on the channel, and implement security by using cryptographic encryption. 

Let us illustrate the problem with an example. 
The following $\pi$-calculus process  describes a
protocol for the exchange of a confidential information:
\[
P= \SENDn s{\textrm{CreditCard}}  \PAR
\RECEIVE sx\IFs {x=\textrm{OwnerCard}} (\SENDn p{\textrm{Ok}} \PAR \SENDn ps) 
\qquad p\ne s
\]
In this specification, the thread on the left sends a
credit card number over the channel~$s$ to the thread on the right which is waiting
for an input on the same channel. If the received card number is the expected one,
then the latter both sends an ack and forwards the communication channel~$s$ on a public
channel~$p$. 
The problem is that, while the confidentiality of the information would require  the context to be
unable to interfere with the protocol and to steal the credit card number,  
in fact this is not guaranteed in the $\pi$-calculus where 
interaction with a parallel process waiting for input on channel~$s$ is allowed.

To amend this problem, the idea is to let the channel for the
exchange of the secret information  available only to the process $P$, 
restricting its scope to $P$ with the  declaration: $\NR s P$. 
The $\pi$-calculus semantics makes the exchange invisible to the context. This is formalized
by the following observational equation stating that no $\pi$-calculus context 
can  tell apart $P$ from its continuation:
\begin{align}
\NR s P\cong^{\textrm{obs}}_\pi \NR s \IFs
{\,\textrm{CreditCard}=\textrm{OwnerCard}}(\SENDn p{\textrm{Ok}} \PAR \SENDn ps) 
\label{eq:inact} 
\end{align}

Unfortunately, to preserve such behavioral
equations when  processes are deployed in  untrusted
environments  is difficult,  since, as explained above, we cannot rely on dedicated
channels for
communication on names created by the {\sf new} operator. 
One natural approach to cope with this problem is to map the private communication 
within the scope of the {\sf new} into  open communications protected by cryptography.

For instance,  the process $\NR s P$  could be implemented in the spi calculus protocol
$\encSQ {\NR sP}$ below by using a public-key crypto-scheme.
In this implementation the creation of a $\pi$-calculus channel $s$  is mapped into the 
creation of a couple
of spi calculus keys: a public key~$s^+$ and a private key~$s^-$.
The receiver performs decryption  of the crypto-packet ${\pack{\textrm{CC}}{s^+}}$
with the private key $s^-$; the operation  assigns the card number to the variable
in the conditional test.  
\begin{align*}
\encSQ {\NR s P} \eqdef& \NR{s^+,s^-}
(\SEND {net}{\pack {\textrm{CC}}{s^+}}\INACT  \PAR   {net}(y).\DECRYPT y{\pack
x{s^-}}Q)
\\
Q\eqdef& \IFs {x=\textrm{OC}}{ \SENDn {net}{\pack{\textrm{Ok}}{p^+}}  \PAR  \SENDn
{net}{\pack{s^+,s^-}{p^+} }} 
\end{align*} 
Unfortunately, the naive protocol above suffers from a number of problems,
among which the most serious is the lack of forward secrecy~\cite{abadi98}: this property 
would guarantee that if keys are corrupted at some time~$t$ then the protocol steps
occurred before~$t$ do preserve secrecy.  
In particular, forward secrecy requires that the content of the packet~$\pack
{\textrm{CC}}{s^+}$, which is the credit card number, is not disclosed if  at some step of
the computation the context gains the decryption key~$s^-$. Stated differently, the
implementation $\encSQ \cdot$  should preserve the semantics of equation (\ref{eq:inact}):
that is, it should be fully abstract. 
It is easy to see that this is not the case since 
a spi calculus context can first buffer the encrypted packet and subsequently, whenever it
enters in posses of the decryption key, retrieve the confidential information; this
breaks equation (\ref{eq:inact}).
While a solution to recover the behavioral theory of $\pi$-calculus is
available~\cite{popl07}, the price to pay is  a complex cryptographic
protocol that relies on a set of trusted authorities acting as proxies.  

Based on  these considerations, in this paper we argue that the
restriction operator of $\pi$-calculus does not adequately ensure confidentiality. 
To tackle this problem, we introduce an operator to
program explicitly secret communications, called {\sf hide}.  
From a programming language point of view, the envisaged use of the
operator is for declaring secret a medium used for {\it local} inter-process
communication; examples include pipelines,  message queues and  IPC
mechanisms of  microkernels. 
The operator is static: that is, we assume that the scope of hidden
channels can not be extruded. The motivation is that all processes using a private channel
shall be included in the scope of its {\sf hide} declaration; processes outside
the scope represent another location, and must not interfere with the protocol. 
Since the {\sf hide} cannot extrude the scope of secret channels, we can use it to directly build specifications
that preserves forward secrecy.
In contrast, we regard the
restriction operator of the $\pi$-calculus, {\sf new}, as useful to create a new  
channel for message passing with scope extrusion, and which does not provide secrecy
guarantees. 

To emphasize the difference between {\sf hide} and {\sf new}, we introduce a
 \emph{spy} context that represents a side-channel attack on the non-dedicated channels.
In practice, 
 \emph{spy} is able to detect whether there has been a communication on one of the
channels not protected by a {\sf hide}, but 
 is not able to retrieve its content.

 \medskip\noindent{\bf {Contributions.}}
 We introduce the  {\it secret} $\pi$-calculus as an \review{orthogonal} extension of the
$\pi$-calculus with  an operator representing confidentiality ({\sf hide}). We 
develop its structural operational semantics and its observational theory.  In particular,
we provide a reduction semantics, a labelled
transition semantics and an observational equivalence.
We show that the observational equivalence induced by the reduction semantics coincides
by the labelled transition system semantics. 
To illustrate  the difference between {\sf hide} and {\sf new}, we shall also consider a
distinguished process context,  called \emph{spy}, representing a side-channel attack.
  
 \medskip\noindent{\bf Plan of the paper}
 In the next section we introduce the syntax and the reduction semantics of
 the secret $\pi$-calculus. 
 In Section~\ref{sec:observational} we present the observational equivalence, and a
characterization based on labelled transition semantics, that we show sound
 and complete.  
 In Section~\ref{sec:spy} we introduce the   {\it spy} process, and we extend
the reduction semantics and bisimulation method accordingly. 
 In Section~\ref{sec:examples} we discuss some algebraic 
 equalities and inequalities of the secret $\pi$-calculus, and we analyze some
interesting examples, notably an
implementation of  name matching,  and a deployment of mandatory access control.  
 Finally, Section~\ref{sec:discussion} presents related work and concludes. 
 An extended  version of the paper containing all proofs  is available
online~\cite{tech-report}.
 
\section{Secret $\pi$-calculus}
\label{sec:pi-calculus}

This section introduces the syntax and the semantics of our  calculus, 
 the {\it secret $\pi$-calculus}.  
The syntax of the processes in Figure~\ref{fig:syntax}  extends  
that of the  $\pi$-calculus~\cite{MPW92a,MPW92b} by:
(1) We consider two binding operators: 
{\sf new }, which -- as we will argue -- does not offer enough security 
guarantees, and {\sf hide}, which serves to program secrecy.
(2) We use two forms of restricted  pattern matching in input, so
that we can  deny a process to receive a (possibly empty) set of channels, or we can
enforce a process to receive only trusted channels. When in the first form the set of
channels is empty we have the standard input of $\pi$-calculus.   
%
\begin{figure}
  \begin{align*} 
    P,Q  \grmeq  & & \text{Processes:}  
    \\
    &  \RECEIVE x{y\forbids{ B}}P & \text{input} &&  \NR{x}(P) &&
\text{restriction}
    \\
    &  \RECEIVET x{y\accept{  A}}P& \text{trusted input}
    && \SR{x}[P] && \text{secrecy} 
    \\
    &  \SEND xyP & \text{output}      && \INACT && \text{inaction} 
    \\
     &P\PAR Q &  \text{composition}   && !P &&
\text{replication} 
  \end{align*}
  \caption{Syntax of the secret $\pi$-calculus}\label{fig:syntax}
  \end{figure}
We use an infinite set of names ${\cal N}$, ranged over by $a,b,\dots,x,y,z$, to represent
channel names and parameters, i.e. the subjects and the objects of communication, 
respectively.  
We let $A,B$ range over subsets of ${\cal N}$. 
A process of the form $\RECEIVE x{y\forbids B }P$ represents an input
where the name $x$ is the input channel name, 
$y$ is a formal parameter which can appear in the continuation $P$, 
and~$B$ is the set of {\it blocked} names that the process cannot receive.
On contrast, an input process of the form $\RECEIVET x{y\accept A }P$ 
declares the object names that the process can  {\it accept}: that is, 
the process accepts in input a name $z$ only if  $z\in A$. 
This permits to program security protocols where only trusted names can be received. 
The free and the bound names of such process are defined as follows: 
$\fv(x[y\forbids B].P)=(\fv(P)\setminus \{y\})\cup\{x\}\cup B$ and
$\bv(x[y\forbids B].P)=\{y\}\cup \bv(P)$, 
$\fv(x(y\accept A).P)=(\fv(P)\setminus \{y\})\cup\{x\}\cup  A $ and
$\bv(x(y\accept A).P)=\{y\}\cup \bv(P)$. 

Processes $\SEND xyP$, $\NR x(P)$, $P\mid Q$, $!P$, and $\INACT$ are the pi calculus
operators respectively describing an output of a name~$y$ over channel~$x$,
restriction of~$x$ in $P$, parallel composition, replication and inaction;
see~\cite{SanWalk01} for more details.  

The process $\SR x[P]$ represents a process $P$ in which the name $x$ is regarded as
\emph{secret}, and should not be accessible to any 
process external to $P$. $\SR x[P]$ binds   the occurrence of $x$ in $P$: 
$\fv(\SR x[P])=\fv(P)\setminus \{x\}$, and $\bv(\SR
x[P])=\{x\}\cup \bv(P)$.

Contexts  are  processes containing a hole~$-$. We write $C[P]$ for
the  process obtained by replacing~$-$ with $P$ in $C[-]$.
\begin{align*} 
C[-]   &\grmeq - \;\mid\;  C[-]\PAR P \;\mid\;  P\PAR C[-]  \;\mid\; \NR x[-]\;\mid\; \SR
x[-]
&& \text{contexts}
\end{align*} 

We  write $x(y).P$ as a short of $x(y\forbids \emptyset).P$, and omit curly brackets in
$x(y\forbids\{b\}).P$ and $x[y\accept\{a\}].P$. When no ambiguity is possible, we will
remove scope parentheses in $\NR x(P)$ and $\SR x[P]$.
We will often avoid to indicate trailing~$\INACT$s.

The combination of the accept and the block  construct  permits to design
processes which are not 
subject to interference attacks from the context. 
We note that their role is dual: the  accept operator prevents the reception (intrusion)
of untrusted names from the environment, and its use is specified by the programmer.  
The block mechanism prevents another process from sending (extruding) a secret name, 
and it is inserted automatically by the system to ensure the protection of such 
names. 
One may wonder whether we could have used just one form of (trusted) input, and declare
the names
to be blocked by accepting all names in $\cal N$ but the intended ones. The main reason
that guided our choice is that we believe that our form of input with blocked names can be
effectively implemented, for instance by using blacklists. Also, we think that there is a
nice symmetry among processes $\RECEIVE x{y\forbids B}P $ and $\NR xP$, and among
processes $\RECEIVET x{y\accept A}P $ and $\SR xP$.

We embed the block mechanism in the rules for structural congruence through the 
operation~$\uplus$ defined in Figure~\ref{fig:reduction}. {\it Blocked} names could indeed
be introduced both statically and dynamically, i.e. when structural congruence is
performed during the computation. We leave the time when the system blocks explicitly the
name in components as an implementation detail. 
Note that in  the second rule of the first line 
the name $b$ is guaranteed to be different from all the names in $A$, because in the
congruence rule for \emph{hide} (cfr same Figure) the free names of Q are
required to be different from the name we want to hide, so the alpha conversion should be
applied .


%
\begin{figure}[!t]
 \emph{Rules for blocking a name}
\begin{gather*}
(\RECEIVE x{y\forbids B}P) \uplus b  \eqdef \RECEIVE x{y\forbids B\cup\{b\}}(P\uplus b)
\qquad
(\RECEIVET x{y\accept A}P) \uplus b  \eqdef \RECEIVET x{y\accept A}(P\uplus b)  
\\
(\NR{x} (P)) \uplus b \eqdef \NR{x} (P\uplus b)^* 
\qquad  
(\SR{x} [P]) \uplus b  \eqdef \SR{x} [P\uplus b]^*  \quad(*)\, b\ne x
\\ 
(\SEND xyP)\uplus b\eqdef \SEND xy{(P\uplus b)}
\qquad 
(P\PAR Q)\uplus b  \eqdef P\uplus b \PAR  Q\uplus b
\\
(! P)\uplus b  \eqdef !(P\uplus b) 
\qquad
\INACT\uplus b\eqdef\INACT
\end{gather*} 
 \emph{Rules for structural congruence}
  \begin{gather*}
    P\PAR Q \equiv Q\PAR P
    \qquad
    (P\PAR Q)\PAR J \equiv P\PAR (Q\PAR J)
    \qquad
    !P \equiv P \PAR !P 
    \\
    \NR x(\INACT ) \equiv \INACT 
     \qquad
       \SR x[\INACT]\equiv \INACT 
    \\
         \NR x (P)\PAR Q \equiv \NR x(P\PAR Q) \quad x\not\in\fv(Q) 
    \\    
         \SR x [P]\PAR Q \equiv \SR x[P\PAR Q\uplus x] \qquad x\not\in\fv(Q)
    \\ 
     \NR x(\SR y [P])\equiv \SR y[\NR x ( P )]  \quad x\ne y
  \end{gather*}
 \emph{Reduction rules} 
 \vspace{-1.3em}
    \begin{gather*} 
    \tag*{\rcom}
    \frac{ 
    z\not\in B}
    {
    \RECEIVE x{y\forbids B}P \PAR \SEND  xzQ 
    \osred
      P\subs zy \PAR Q
    }
    \\
        \tag*{\rtcom}
    \frac{
    z\in A
    }
    {
    \RECEIVET x{y\accept A}P \PAR \SEND  xzQ  
    \osred
      P\subs zy  \PAR Q
    }
    \\
    \tag*{\rres,\rhide}
    \frac{
    P \osred P'  
    }{
    \NR{x}(P)  \osred\NR{x}(P')
    } 
    \qquad
     \frac{
      P \osred  P' 
    }{
     \SR{x}[P] \osred\SR{x}[P']
    }
    \\ 
   \tag*{\rpar,\rstruct}
   \frac{
      P \osred  P'
    }{
      P \PAR Q \osred P' \PAR Q
    }
    \qquad
       \frac{
      P \equiv Q \quad Q \osred Q'
     \quad Q' \equiv P'
    }{
      P \osred P'
    } 
\end{gather*}
\caption{Reduction semantics} 
\label{fig:reduction}
\end{figure}

Following standard lines, we define the semantics of our calculus
via a reduction
relation, also specified in Figure~\ref{fig:reduction}.
We assume a capture-free substitution operation $\subs zy$: 
the process $P \subs zy$ is obtained from $P $ by substituting all the free occurrences of $y$ 
 by $z$.
As usual, we use a structural congruence $\equiv$ to rearrange processes.
Such congruence includes the equivalence induced by  alpha-conversion, and the relations defined in  Figure~\ref{fig:reduction}.
The rules for the $\pi$-calculus operators (first line) are the standard ones.
The rules for inaction under a binder follow (second line).
We recall that the scope extrusion rule for  \keyword{new} (third line) permits to
enlarge the scope of a name and let a process receive it. 
In contrast, the scope extrusion rule for  \keyword{hide} (fourth line) permits to 
enlarge the scope of a name, but at the same time it sets the name to \emph{blocked} for
the process which are being included in the scope, thus preventing them  to receive the
name. The last rule (fifth line) permits to swap the two binders. 
  
The first rule for reduction, \rcom, says that an input process of the form
 $\RECEIVE x{y\forbids B}P$ is allowed to synchronize with an output process $\SEND
xzQ$ and receive the name~$z$ provided that~$z$ is not \emph{blocked}
($z\not\in B$). 
The result of the synchronization is the progression of both the receiver and the sender,
where the formal parameter in the input's continuation is replaced by the name~$z$.  Note
that whenever $B=\emptyset$ we have the standard communication rule of the $\pi$-calculus.
The main novelty is represented by the rule for trusted communication 
 \rtcom. This rule says  that an output process can send a
name $z$ over $x$ to a parallel process waiting for input on $x$, provided that $z$ is
explicitly declared as accepted ($z\in A$) by the receiver. If this is the case, the name
will  replace the occurence of the formal parameter in the input's continuation.  
Rules \rres and \rhide are  for   {\sf new} and for  
{\sf hide}  respectively,  and follow  the same schema.  
The rules for parallel composition, replication and incorporating structural congruence
are standard. 

We let $P \Rightarrow P'$ whenever either (a) $P \rightarrow \cdots \rightarrow P'$,  or 
(b) $P'=P$.

\begin{example}\label{example:secrecy}
We show  how   {\sf hide} can be used to prevent the extrusion of a secret.
Consider the                                                   
process:\vspace{-.5em}
\[ 
P \eqdef  \SR {z}[ \SENDn xv ]      \qquad x\ne z
\]
The process $\SENDn xv$ can be interpreted as an internal attacker trying to
leak the name $v$ to a context  $C[-]\eqdef -\PAR x(y).\SENDn{leak}y $. 
By using the  structural rule for enlarging the scope of hide in
Figure~\ref{fig:reduction} we infer that $C[P] \equiv \SR {z}[ \SENDn xv \PAR
x(y\forbids
z).\SENDn{leak}y]$.
Whenever the name $v$ is not declared secret, that is whenever $v\ne z$, the leak
cannot be prevented: by applying \rcom,\rhide, and \rstruct we have $C[P]\osred
\SENDn{leak}v $.
Conversely, when the name $v$ is protected by {\sf hide}, that is $v=z$,  we do
not have any interaction  and  secrecy is preserved.
\end{example} 

\begin{example}\label{example:accepted}
The combined use of  the accept and block sets permits to avoid
interference with the context. 
Consider the process below,  where $n>0$:\vspace{-.6em}
\begin{align*}
P &\eqdef \SR {z_1} \cdots \SR {z_n}[\cdots[x[y\accept Z].P \PAR \SENDn
x{z_i} ]\cdots] 
&&
Z \subseteq \{z_1,\cdots z_n\},i\in \{1,\dots,n\}  
\end{align*}
Take a context $C[-] \eqdef -\PAR \NR y  !\SENDn xy   \PAR 
!x(w)$.  Such context is unable to send the fresh name $y$ to $P$,
because the input process in $P$  is programmed to accept only  trusted names protected
by {\sf hide}. Dually, the context  cannot
receive the  protected name~$z_i$.  
Therefore $C$ and~$P$ cannot interact: $C[P]\osred Q$ implies that a)  
$Q\equiv C[\SR {z_1} \cdots \SR
{z_n}[\cdots[P\subs {z_i}y ]\cdots]]$  or b) $Q\equiv C[P]$.
\end{example}

\section{Observational equivalence}
\label{sec:observational}

In this section we define a notion of behavioral equivalence based on observables, or
barbs. As the reader will notice, a distinctive feature of our observational theory is
that trusted inputs are visible only under certain conditions, namely that the context
knows at least a name that is declared as accepted. Conversely, processes trying to send
a name protected by an hide declaration are not visible at all.
The choice to work in a synchronous setting
permits us to emphasize the differences among our theory and that of
$\pi$-calculus. However, the same results would hold for a secret asynchronous
$\pi$-calculus, while the contrast would be less explicit as input barbs would not be
observable.

We say that a name $x$ is  bound in $P$ if $x\in\bv(P)$.
An occurrence of   $y$ is  hidden in $P$ if such occurrence of $y$ appears
in the scope of a {\sf hide} operator in~$P$.

\begin{definition}[Barbs] 
We define: 
\begin{itemize}
  \item $P\barb_{x}$ whenever $P\equiv C[\RECEIVET x{y\accept A} Q]$  with $x$ not
bound in $P$ and $A\cap\bv(P)\ne A$, or whenever 
  $P\equiv C[\RECEIVE x{y\forbids B} Q]$ with $x$ not bound in $P$. 
 \item $P\barb_{\overline x}$ whenever    
 $P\equiv C[\SEND xyQ]$ with $x$ not bound in $P$ and~$y$ not hidden in~$P$.
\end{itemize}
\end{definition} 

\noindent
Based on this definition,  we have that  
$P_1\eqdef\SR x \RECEIVET z{y\accept x}Q$,    
$P_2\eqdef\NR x \RECEIVE x{y\forbids B}Q$, and
$P_3\eqdef \RECEIVET z{y\accept \emptyset}Q$ do  not exhibit a barb $z$, written
$P_i \nbarb_z$ for  $i=1,2,3$.  
In contrast, when $x\ne z$ and $A\cap \{x\}\ne\emptyset$ 
we have  that $\NR x \RECEIVET
z{y\accept
A}P \barb_z$, and when $x\ne z$ we have $\SR x \RECEIVE z{y\forbids B}P$. 
Whenever $P\eqdef \SR y \SEND xvQ$ with $y\ne x$, we have $P\barb_{\overline x}$ if
$y\ne v$, and $P\nbarb_{\overline x}$ otherwise.
Weak barbs are defined by ignoring reductions. 
We let  $P\Barb_{x}$
whenever $P\Rightarrow P'$ and $P'\barb_{x}$; similarly $ P\Barb_{\overline x}$ whenever $
P\Rightarrow P'$ and $P'\barb_{\overline x }$. 

Following the standard definition of observational equivalence, we are aiming at an 
equivalence relation that is sensitive to the barbs,  is closed under reduction, and is
preserved by certain contexts.      

\begin{definition}[Barb preservation]
A  relation $\CR$ over processes is barb preserving if 
$P\CR Q$, $P\barb_{x}$ implies $Q\Barb_{x}$, and 
  $P\barb_{\overline x}$ implies $Q \Barb_{\overline x}$.
\end{definition}

The requirement of reduction closure is to ensure that the processes maintain their
correspondence through the computation. 

\begin{definition}[Reduction closure]
A  relation $\CR$ over processes is reduction-closed if $P\CR Q$
and $P\rightarrow P'$ implies that $Q\Rightarrow Q'$ and 
$P'\CR Q'$.
\end{definition}

We require contextuality  with respect to the parallel composition,  the new and  the
hide operators (cf. Section~\ref{sec:pi-calculus}).  

\begin{definition}[Contextuality]
A relation $\CR$ over processes is contextual if $P\CR Q$ implies
$C[P] \CR C[Q]$. 
\end{definition}

\begin{definition}[Observational equivalence]\label{def:obs-equivalence}
Observational equivalence, noted $\cong$, is the largest symmetric  
relation over processes  which is barb preserving, reduction closed
and contextual.  
\end{definition} 

Observational equivalence is difficult to establish since it requires
quantification over contexts. In the next section we will introduce labelled transition
semantics for the secret $\pi$-calculus, and show that the induced bisimulation coincides
with observational equivalence. Besides the theoretical interest, this will be also
of help in  proving that two processes are observationally equivalent. 
 
\subsection{Characterization} 
\label{sec:bisimulation}
\begin{figure}[t] 
\begin{gather*}
\tag*{\linp,\linpt}
\frac{
z\not\in B
}
{
 x(y\forbids B).P\lts{ x(z)} P\subs zy 
}
\qquad
\frac{
z\in A
}
{
 \RECEIVET x{y\accept A} P \lts{ x(z)} P\subs zy 
} 
\\[2mm]
\tag*{\lout,\lopen}
\frac{}
{
 \SEND xyP \lts{ \SENDn xy} P 
}\qquad
\frac{
P \lts{ \SENDn xy} P'
\qquad
y\ne x 
}{
\NR y  P \lts{ (y)\SENDn xy} P'
}
\\[2mm]
\tag*{\lcom}
\frac{
P\lts{ x(y)} P'
\qquad
Q \lts{ \SENDn xy} Q'
}{
P\PAR Q\lts{ \tau}  P'\PAR Q' 
}
\\[2mm]
\tag*{\lclose}
\frac{
P\lts{ x(y)} P'
\qquad
Q \lts{ (y)\SENDn xy} Q'
\qquad
y\not\in\fv(P)  
}{
P\PAR Q\lts{ \tau}  \NR{y}(P'\PAR Q')
}
\\[2mm]
\tag*{\lres,\lhide}
\frac{
P \lts{ \alpha} P'
\qquad
x\not\in\fv(\alpha)
}{
\NR x P \lts{ \alpha} \NR x P'
} 
\qquad
\frac{
P \lts{ \alpha}  P'
\qquad
x\not\in \fv(\alpha)
}{
\SR x P \lts{ \alpha} \SR x P'
} 
\\[2mm]  
\tag*{\lpar,\lrepl}
\frac{
P\lts{ \alpha} P'
\qquad
\bv(\alpha)\cap\fv(Q)=\emptyset
}{
P\PAR Q \lts{ \alpha}P'\PAR Q
}
\qquad
\frac{
P\lts\alpha P'
  }{
P\lts{ \alpha}   P'\PAR  !P  
}
\end{gather*}
\caption{Labelled transition system}
\label{fig:lts}
\end{figure}

The characterization relies on labelled transitions of the form $P\lts\alpha P'$, 
where $\alpha$ is one of
the following actions:
\[
\alpha =   x(z)\mid  \SENDn xz  \mid  (z)\SENDn xz 
\mid \tau  
\]
We let $\fv(x(z))=\{x\}$, $\fv(\SENDn xz)=\{x,z\}$, and 
$\fv{(z)\SENDn xz}=\{x\}$. 
We define $\bv(x(z))=\{z\}$, $\bv(\SENDn xz)=\emptyset$ and
$\bv( (z)\SENDn xz)=\{z\}$. We let $\fv(\tau)=\emptyset=\bv(\tau)$.

The transitions are defined by the rules in Figure~\ref{fig:lts}. 
Action $ x(z)$ represents the receiving of a name $z$ on a
channel $x$. In rule \linp, a process
of the form $x(y\forbids B).P$ can receive a value $z$ over $x$,  provided that $z$ is is
not blocked ($z\not\in B$). The received name will
replace the formal parameter in the body of the continuation. 
Rule \linpt describes a trusted input, that is a process
of the form $\RECEIVET x{y\accept A}P$ that receives a variable $z$ over $x$ whenever 
$z$ is accepted ($z\in A$); the variable~$z$ will replace all occurrences of~$y$ in $P$.
The action $ \SENDn xy$ represents the output of a   name $y$ over
$x$. This move is performed in \lout by the process $\SEND xyP$ and leads to the
continuation $P\rhd B$. 
Communication arises in rule \lcom by means of a $\tau$ action obtained by a
synchronization of an $ x(y)$ action with a $ \SENDn xy$ action.
Action $ (y)\SENDn xy$ is fired when the name $y$ sent over $x$ is bound
by the {\sf new } operator and its scope is opened by using  rule \lopen. 
The scope of the {\sf new } is closed by using  rule  \lclose.
In this rule the scope of a name $y$ sent over $x$  is  
enlarged to include a process which executes a dual action $ x(y)$, 
giving rise to a synchronization of the two threads depicted by an action $\tau$. 
Rule \lres is standard for restriction. 
Rule \lhide says that process $\SR x P $  performs an action $\alpha$ inferred
from $P$, provided that the $\alpha$ does not contain $x$.
Therefore extrusion of hidden channels is not possible, as previously discussed; note
indeed that this the unique rule applicable for \emph{hide}. 
Rule \lrepl performs a replication. 
 
%
We have a standard notion of bisimilarity; in the following, we let $\Lts{\tau}$ be
the reflexive and transitive closure of $\lts{\tau}$.

\begin{definition}[Bisimilarity]\label{def:bisimulation}
A symmetric relation $\CR$  over processes is a  bisimulation if
whenever $P\CR Q$  and $P\lts{\alpha} P'$ then there exists 
a process $Q'$ such that 
$Q \Lts{\tau}\lts{\hat\alpha}\Lts{\tau}Q'$ and 
$P'\CR Q'$  where $\hat\tau$ is the empty string and
$\hat\alpha=\alpha$ otherwise.  
Bisimilarity, noted $\approx$, is the largest bisimulation.
\end{definition}

The following result establishes that bisimilarity can be used as a proof technique for
observational equivalence; the proof is by coinduction and relies on the closure of
bisimilarity under the {\sf new}, {\sf hide} and parallel composition operators. 
 
\begin{proposition}[Soundness]\label{prop:soundness}
If $P\approx Q$ then  $P\cong Q$. 
\end{proposition}

To prove the reverse direction, namely that behaviourally equivalent processes are
bisimilar, we follow the approach of Hennessy~\cite{Hen07} and proceed by co-induction
relying on contexts $C_\alpha$ which emit the desired barbs whenever they interact
with a process $P$ such that $P\Lts{\alpha} P'$, and vice versa. 
Perhaps interestingly, we can program a context to check if
a given name is fresh even if our syntax does not include  a matching construct 
(cf.~\cite{Hen07,BorealeS98}).
In Section~\ref{sec:examples} we will show that in the secret $\pi$-calculus the process 
$\IF{x=y}PQ$ can be derived. 

\begin{proposition}[Completeness]
\label{prop:completeness}
If $P\cong Q$ then $P\approx Q$. 
\end{proposition} 
\begin{proof}
Let $P\CR Q$ whenever  $P\cong Q$ and assume that
$P\lts\alpha P'$. We show that there is $Q'$ such that $Q\Lts{\hat\alpha} Q'$
and $P'\equiv \CR \equiv Q'$; this suffices to prove that $\CR$ is included in
observational equivalence (cf.~\cite{SangiorgiM92}).
Whenever $\alpha=\tau$, we use reduction-closure of $\cong$ to find $Q'$ such that
$Q\Lts{} Q'$ with $P'\cong Q'$. By relyng on a lemma that establishes that
reductions correspond to $\tau$ actions, we infer that $Q\Lts{\tau} Q'$, which is the
desired result since $P'\CR Q'$. Otherwise assume $\alpha\ne\tau$.
We exploit contextuality of $\cong$ and infer that
$C^A_\alpha[P]\cong C^A_\alpha[Q]$ where we let $A=\fv(P)\cup\fv(Q)$ and $C^A_\alpha$ be
defined below. We let $A=\{a_1,\dots,a_n\}$, $I=1,\dots,n$, with $n\geq 1$, and
assume names $\omega,\psi_1,\dots,\psi_n$ such that
$\{\omega,\psi_1,\dots,\psi_n\}\cap A=\emptyset$.  
\begin{align*}
C^A_{x(y)}[-] &\eqdef  - \PAR \SEND xy\SENDn \omega{}  \\
C^A_\alpha[-] &\eqdef  
- \PAR \SR k[x(z).( z[w\accept k]  \PAR \SENDn \omega{}) 
\PAR_{i\in I} \SEND  {a_i}k\SENDn{\psi_i}{}]
&& \alpha=x\oput y,(y)x\oput y\\
C'_y&\eqdef \SR k [y[w\accept k]\PAR \SENDn \omega{}\PAR_{i\in I} \SEND 
{a_i}k\SENDn{\psi_i}{}] &&\forall i\in I\,.\, y\ne a_i 
\\
C''_y &\eqdef \SR k [\SENDn \omega{} \PAR
\SENDn{\psi_l}{}\PAR_{i\in I\less l}\SEND{a_i}k\SENDn{\psi_i}{} ]
&&a_l=y
\end{align*} 
Assume $\alpha=\SENDn xy$. We have that there is $a_l\in A, a_l=y$ such that
$C^A_\alpha[P]\Lts{}\equiv C_P\eqdef P'\PAR C''_y$. We find
a process $C_Q$ such that
$C^A_\alpha[Q]\Lts{}C_Q\cong C_P $.
Since $C_P\barb_{\bar \omega}, \barb_{\bar \psi_l}$, this implies that  
$C_Q\Barb_{\bar \omega}, \Barb_{\bar \psi_l}$. Therefore the weak barb  
$\bar \omega$  of $C_Q$  has been unblocked since $Q$ emits a weak action $\alpha'$ with
$x$ as subject. Moreover, the object of $\alpha'$ is~$y$, that is $\alpha'=\alpha$,
because of the weak barb  $\bar \psi_l$. Indeed the thread $\SEND {a_l}k\SENDn
{\psi_l}{}$ with $a_l=y$ can be unblocked only by $y[w:k]$, because $k$ is protected by
the {\sf hide} declaration. 
Therefore there is $Q'$ such that $Q\Lts{\alpha} Q'$ and $C_Q\cong Q'\PAR C''_y$.
We conclude by showing that this implies $P'\cong Q'$, and in turn $P'\CR Q'$, as
requested. 
Assume $\alpha=(y)x\oput y$. We have that $C^A_\alpha[P]\Lts{}\equiv C_P\eqdef P'\PAR
C'_y$. Since $y$ is fresh we have that  $a_i\ne y$ for all $a_i\in A$. Therefore
$C_P\nBarb_{\bar \psi_i}$ for all $i\in I$, because $k$ is protected by {\sf hide}.
We easily obtain that there is $C_Q$ such that $C^A_\alpha[Q]\Lts{} C_Q\cong C_P$ with
$C_Q\Barb_{\bar\omega},\nBarb_{\bar \psi_i}$, for all $i\in I$. This let us infer that 
there is $Q'$ such that $Q\Lts{\alpha} Q'$ and $C_Q\cong Q'\PAR C'_y$, and the result then
follows by showing that $P'\PAR C'_y\cong Q'\PAR C'_y$ implies $P'\cong Q'$.
\end{proof} 

Full abstraction is obtained by
Propositions~\ref{prop:soundness} and~\ref{prop:completeness}.

\begin{theorem}[Full Abstraction]
\label{theor:character}
$\cong\ =\ \approx$. 
\end{theorem} 
\section{Distrusting communications protected by restriction}
\label{sec:spy} 

In this section we introduce a {\em spy} process that 
represents a side-channel attack against communications that  occur on untrusted
channels, that is: channels that are not protected by  {\sf hide}. 
We assume that the spy is not able to retrieve  the content of an exchange.
The spy abstraction models the ability of the context to detect  interactions
when the processes are implemented  by means of  network protocols which do not rely on 
dedicated channels, and therefore require some mechanism to enforce the secrecy of the
message (e.g. cryptography). 
This ability  leads to break some of the standard security
equations for 
the {\sf new} operator, which can be recovered by re-programming the protocol and making  
use of the {\sf hide} operator.
We  add to the syntax of the secret $\pi$-calculus the following process where we
let \keyword{spy} be a reserved keyword. We let $P,Q,R$ to range over \emph{spied
processes}. 

  \begin{align*}  
  P,Q,R   &\grmeq  \cdots \;\mid\;  \spyb P && \text{spied processes}\\
  S &\grmeq \{x\} \;\mid\; \emptyset            && \text{spied set}
  \end{align*}

   \begin{figure}[!t]
    \emph{New rules for blocking a name}
    \begin{gather*} 
    (\spyb P)\uplus b\equiv \spyb (P\uplus b)  
    \end{gather*}
    \emph{New rules for structural congruence}
    \begin{gather*} 
     \NR x (P)\PAR \keyword{spy}.R  \equiv 
    \NR x(P\PAR \keyword{spy}\accept  x .R) \qquad x\not\in\fv(\keyword{spy}.R)
    \\
     \NR x (P)\PAR \keyword{spy}\accept  y.R  \equiv 
    \NR x(P\PAR \keyword{spy}\accept  y .R) \qquad x\not\in\fv(\keyword{spy}\accept  y .R)
    \\
     \SR x [P]\PAR \spyb R \equiv 
    \SR x[P\PAR (\spyb R)\uplus x] \qquad x\not\in\fv(\spyb R)
    \end{gather*}
    \emph{New reduction rules}  
    \begin{gather*} 
    \tag*{\rscom}
    \frac{ 
    z\not\in B}
    {
    \RECEIVE x{y\forbids B}P \PAR \SEND  xzQ \PAR \keyword{spy}\accept x .R
    \osred
      P\subs zy \PAR Q\PAR  R
    }
    \\
        \tag*{\rstcom}
    \frac{
    z\in A 
    }
    {
    \RECEIVET x{y\accept A}P \PAR \SEND  xzQ  \PAR \keyword{spy}\accept x.R
    \osred
      P\subs zy  \PAR Q \PAR R
    }
    \end{gather*}
    \caption{Spied process  semantics}\label{fig:spy}
    \end{figure}

When in $\spyb P$ the spied set $S$ is equal to $\{x\}$, noted
$\keyword{spy}\accept x . P$, this  permits to make explicit which (free) reduction
the spy shall observe. Note that listening on multiple names can be easily programmed by
putting in parallel several spies. 
The spy process $\keyword{spy}\accept\emptyset. P$, noted $\spy P$,  will be used to
detect reductions protected by restriction. 
We let the free and bound names of
the
\emph{spy} be defined as follows:
$\fv(\spyb R)\eqdef S\cup\fv(R)$ and 
$\bv(\spyb R)\eqdef \bv(R)$. 

The semantics of spied processes is
described by adding the communication rules in Figure~\ref{fig:spy} to those in
Figure~\ref{fig:reduction}: 
The rules describe  a form of synchronization among three processes: a sender on
channel~$x$, a receiver on channel~$x$, and a {\it spy} on channel~$x$. More in detail,
rule \rscom depicts a synchronization among an input
of the form $\RECEIVE x{y\forbids B}P$, a sender and a spy, while rule \rstcom describes
a similar three-synchronization but for a trusted input of the form $\RECEIVET x{y\accept 
A}P$. 

The definition of observational equivalence 
for spied processes is obtained by extending  Definition~\ref{def:obs-equivalence}
to the semantics in Figure~\ref{fig:spy}; we indicate the resulting equivalence with
$\scong$. This will permit to study the security of  processes in presence of the
\emph{spy}.

\begin{figure}[!t] 
\begin{gather*}
 \tag*{\slspy,\slspyn}
\frac{  
}
{
\keyword{spy}\accept x.P \lts{ ?x} P 
}\qquad
\frac{  
}
{
\spy P \lts{ ?\nu} P 
}
\\[2mm]
\tag*{\slcom,\slclose}
\frac{
P\lts{ x(y)} P'
\qquad
Q \lts{ \SENDn xy} Q'
}{
P\PAR Q\lts{ !x}  P'\PAR Q' 
}
\qquad 
\frac{
P\lts{ x(y)} P'
\qquad
Q \lts{ (y)\SENDn xy} Q'
\qquad
y\not\in\fv(P)  
}{
P\PAR Q\lts{ !x}  \NR{y}(P'\PAR Q')
} 
\\[2mm]
\tag*{\slspycom}
\frac{
P\lts{!x} P'
\qquad 
Q\lts{?x} Q'
}
{
P\PAR Q\lts \tau P'\PAR Q'
}  
\\[2mm]
\tag*{\slresn,\slhidet}
\frac{ 
    P \lts{\alpha} P'\qquad x\not\in\subject(\alpha) }
{
    \NR x P \lts{ \enc{\alpha}_x} \NR x P'
}\qquad
\frac{
P \lts{ !x} P' \qquad x\not\in\subject(\alpha)\cup\object(\alpha)
}{
\SR x P \lts{\encH{\alpha}_x} \SR x P'
} 
\\[2mm]
\end{gather*}
\caption{Labelled transitions for spied processes}
\label{fig:lts-spy}
\end{figure} 
To make  the picture clear, in Figure~\ref{fig:lts-spy} we introduce labelled
transition semantics for spied processes. 
We consider two new actions $?x$ and $!x$  corresponding
respectively to the presence of a \emph{spy} and  to a signal of communication.
\begin{align*}
\alpha &\grmeq     \cdots \mid\, ?x \mid !x 
\end{align*} 

We assume the existence of variable~$\nu\in\cal N$
that cannot occur in the process syntax,
and we use it to signal restricted communications. 
It is convenient to define the notion of (free) subject and object of an action.
We let $\subject(\alpha)\eqdef \{x\}$ whenever $\alpha=\SENDn xy,(y)\SENDn xy, x(y)$, and
be empty otherwise. We  define $\object(\alpha)\eqdef \{ y\}$ whenever $\alpha=\SENDn
xy,x(y)$, and $\object(\alpha)=\emptyset$ otherwise.

The lts in Figure~\ref{fig:lts-spy} introduces three new rules for the
\emph{spy}, \slspy, \slspyn and \slspycom, and re-defines  the rules for restriction, for
{\sf hide} and for communication of Figure~\ref{fig:lts}. 
In rule \slspy  the process $\keyword{spy}:x. P$ can fire an action $?x$
and progress to~$P$. %
The dual action, $!x$, is fired in rules \slcom and \slclose whenever a communication
occurred on a free channel~$x$.  
Rule \slspycom describes the eaves-dropping of a communication. 
A  process of the form $\keyword{spy}.P$ can only fire an action $?\nu$ through
rule \slspyn.
In rule \lres	 we use a partial  function $\enc{\cdot}_x$ to relabel the action fired
underneath a restriction: we let
$\enc{\alpha}_x\eqdef\alpha$ whenever $x\not\in\fv(\alpha)$, 
$\enc{!x}_x\eqdef !\nu$, $\enc{?x}_x\eqdef ?\nu$.  This will be
used to signal restricted communications, as introduced. 
Differently, in rule \lhide we use a relabeling partial  function $\encH{\cdot}_x$ that
makes  invisible communications that occur under \emph{hide}. We let
$\encH{\alpha}_x\eqdef\alpha$ whenever $x\not\in\fv(\alpha)$, 
$\encH{!x}_x\eqdef \tau$ and $\encH{?x}_x\eqdef \tau$.

\begin{definition}[Bisimilarity]
A symmetric relation $\CR$  over spied processes is a  bisimulation if
whenever $R_1\CR R_2$  and $R_1\lts{\alpha}R'$ then there exists 
a spied process $R''$ such that $R_2 \Lts{\tau}\lts{\hat\alpha}\Lts{\tau}R''$ and 
$R'\CR R''$  where 
$\hat\tau$ is the empty string, and
$\hat\alpha=\alpha$ otherwise. 
Bisimilarity, noted $\sapprox$, is the largest bisimulation.
\end{definition} 

By using the same construction of Section~\ref{sec:bisimulation}, we obtain the
main result of this section: observational equivalence for spied processes and  
bisimilarity  coincide. As a by-product, we can also use bisimulation as a technique to
prove that two processes cannot be distinguished by the {\it spy}. 

\begin{theorem}[Full Abstraction]
\label{theor:spy-fa}
$\scong \, =\,\sapprox$. 
\end{theorem}
\begin{proof}[Sketch of the proof]
To see that behavioural equivalence is included in bisimilarity, we proceed by
co-induction as in the  proof of Proposition~\ref{prop:completeness} by 
relying on contexts $C^A_\alpha$ that detect whenever a process does emit a weak
action $\alpha$.  Given a set of names $A$ such that $\fv(\alpha)\subseteq A$
and $\omega\not\in A$ we define the following contexts to account for the new actions
$!x$ and $?x$.
\begin{align*}
C^A_{!x}[-]&\eqdef \keyword{spy}\accept x. {\SENDn\omega{}} &&x\ne \nu
\\
C^A_{?x}[-]&\eqdef x(y).\SENDn\omega{} \PAR \SENDn x{} &&x\ne \nu
\\
C^A_{!\nu}[-]&\eqdef \keyword{spy} . {\SENDn\omega{}}  
\\
C^A_{?\nu}[-]&\eqdef \NR{x}(x(y).\SENDn\omega{} \PAR \SENDn x{})  
\end{align*}
The proof then proceeds routinely by following a schema similar to the one of
Proposition~\ref{prop:completeness}.
The reverse direction, namely that bisimilarity is contained in behavioural equivalence,
is shown by proving that $\sapprox$ is closed under the {\sf new}, {\sf hide}, and
parallel composition operators. See~\cite{tech-report} for all the details.
\end{proof}

\section{Properties of the secret $\pi$-calculus}
\label{sec:examples} 
In this section we discuss some algebraic 
properties of the secret $\pi$-calculus, and we  
 show how we can implement the name matching operator. Lastly we provide an example
of deployment of a mandatory access control policy that is inspired by the D-Bus
technology~\cite{dbus}.
In the following, we write $P\not\scong Q$ to indicate that $(P,Q)\not\in\;\scong$.
We also write $\SENDn{x}{}$ and omit to indicate the message in output whenever
this is irrelevant, and use the notation $\SR B P$ to indicate the process
$\SR {b_1} \cdots\SR {b_n} P$ whenever $B=\{b_1,\dots,b_n\}$.

\paragraph{\bf Algebraic equalities and inequalities} 
The first inequality illustrates the mechanism of blocked names.
\begin{align}
x(y\forbids B).P \not\scong x(y\forbids B').P  &&
B\ne B'\label{eq:forbids}
\end{align}
To prove (\ref{eq:forbids}) let $z\in B'$, $z\not\in B$ and consider the context 
  $C[-]\eqdef  \SR {B,B'}[\SEND  xz\SENDn\omega{ }\PAR - ]$  with
$\omega$ free, $\omega\not\in\fv(P)$. 
By applying \rcom followed by applications of \rhide we 
have that $C[x(y\forbids B).P ]\osred \SR {B,B'}[\SENDn\omega{ }\PAR P\subs zy] $, that
is 
$C[x(y  ).P]\Barb_{\bar\omega}$. In contrast, we have that  
\mbox{$C[x(y\forbids B').P]\nBarb_{\bar\omega}$}, because of $z\in B'$. 
The case  $B'\subseteq B$ is analogous.

We have a similar result for accepted names.
\begin{align}
x[y\accept A].P \not\scong x[y\accept A'].P  &&
A\ne A'\label{eq:forbids}
\end{align}
A distinguishing context is $C[-]\eqdef   \SEND  xa\SENDn\omega{ }\PAR -  $ 
where $\omega$ is fresh and $a\in A, a\not\in A'$ if $A\not\subset  A'$, and  
$a\in A', a\not\in A$ otherwise.

The next  inequality  illustrates the discriminating power of the
{\it spy}. 
\begin{align} 
\NR x (\SENDn xz \PAR x(y) )&\not\scong \INACT
\label{eq:weak}
\end{align}
To prove (\ref{eq:weak}), consider the context $C[-]=\keyword{spy} .
{\SENDn \omega{}}\PAR -$.
By applying \rscom and \rres followed by \rstruct we infer 
$C[\NR x (\SENDn xy \PAR x(y))]\osred  \SENDn\omega{}$: that is,
$C[\NR x (\SENDn xy \PAR x(y))]\Barb_{\bar{\omega}}$
while $C[\INACT]\nBarb_{\bar \omega}$. 

The invisibility of communications protected 
by using the \emph{hide} operator is established by means of the equation below, which is 
proved
by co-induction.
\begin{align} 
\SR x [\SENDn xz \PAR x(y).Q] &\scong\SR x[Q\subs zy] \label{eq:hide-weak}
\end{align}

The last equation states the impossibility of extrusion of hidden channels. 
\begin{align}
\SR x [ \SENDn zx ]\scong \INACT \label{eq:strong-fs}
\end{align}

\paragraph{\bf Implementing name matching}  
Name matching is not needed as an operator in our calculus
(cf.~\cite{CarboneM03}). We show this  
by  providing a semantics-preserving translation of the if-then-else
construct~\cite{Hen07}. 
Consider the process
$\IF {x=y}PQ$ which
reduces to $P$ whenever $x=y$, and reduces to $Q$ otherwise.
Let $Z\eqdef\fv(\IF
{x=y}PQ)$; therefore there are names $z_1,\dots, z_n$, $n\geq 0$, s.t.   
$Z=\{x,z_1,\dots,z_n\}$. Let $I=\{1,\dots,n\}$ and assume~$k$ fresh. We define:
\begin{align*}
\encSQA{\IF {x=y}PQ} &\eqdef
 \SR k[y[w\accept k] \PAR \SEND xk(P\uplus k)\PAR_{I}\SEND {z_i}k(Q\uplus
k) ]
\end{align*}  
Whenever $x=y$, we have that the only possible reduction arises among
the trusted input $y[w\accept k]$ and $\SEND xk(P\uplus k)$, leading to
$P'\eqdef \SR k[P\uplus k \PAR_I \SEND
{z_i}k(Q\uplus k)]$. Note that  $P$ and $P'$ have the same interactions with
the context, because $k$ is blocked in all threads of $P'$: therefore $Q$ cannot be
unblocked. 
This result can be formalized by relying on the behavioural theory~\footnote{Note that
observational equivalence is not
preserved by input-prefixing; the outlined translation could be indeed sensitive to
name aliasing.} of the secret
$\pi$-calculus.

\noindent
We infer the following equation:
\begin{equation}
\encSQA{\IF {x=x}PQ}\cong P \label{eq:match}
\end{equation}
  
\medskip\noindent
Consider now the case $x\ne y$ and let $y=z_1$. 
The matching process reduces to the rearranged process 
$\SR k[\SEND xk(P\uplus k) \PAR Q\uplus k \PAR_{ 
\{2,\dots,n\}}\SENDn{z_i}k(Q\uplus k)]$, which has the same behaviour
of~$Q$:
%
\begin{equation}
\encSQA{\IF {x=y}PQ}\cong Q \qquad x\ne y \label{eq:mismatch}
\end{equation} 

\paragraph{\bf Modeling dedicated channels}
Security mechanisms based on dedicated
channels can be naturally modeled in the secret $\pi$-calculus. 
D-Bus~\cite{dbus} is an IPC system for software applications 
that is used in many desktop environments. Applications of each user
share a private bus for asynchronous message-passing
communication; 
a system bus
permits to broadcast messages among applications of different users. 
Versions smaller than $0.36$ contain  an erroneous access policy for
channels which allows users to send and listen to messages on another user's channel
if the address of the socket is known. We model this vulnerability  
by means of an {\it internal } attacker that leaks the user's
channel. In the specification below, two applications of an user~$U_1$ utilize
a private bus to exchange a password; in fact, the password  can be intercepted by the
user~$U_2$ through  the malicious
code~$!\SENDn {\sys}{c}$ of $U_1$, which publishes~$c$ on the system  bus. 
    \begin{align}
     U_1&\eqdef \NR{c}(! \SENDn {\sys}{c}\PAR 
\NR{\pwd}\SENDn {c}{\pwd}  \PAR \RECEIVE {c}x{P}  )  
&  U_2&\eqdef \sys(x).x(y_{pwd}).Q  
\end{align} 

The patch released by  Fedora  restricts the access to the
user's bus: only applications with the same user-id  can
have access. We stress that this policy is mandatory: that is, the user cannot
change it. By using the secret \mbox{$\pi$-calculus} we can easily patch $U_1$ by hiding
the
bus: 
$U' \eqdef \SR{c}[! \SENDn {\sys}{c}\PAR 
\NR{\pwd}(\SENDn {c}{\pwd} ) \PAR \RECEIVE {c}x{P}]$. 
The following equation, which can be proved co-inductively,  states that the
policy is fulfilled even in presence
of internal attacks: 
\begin{align} 
U' &\scong \SR{c}[\NR{\pwd}({P}\subs{\pwd} x)] \label{eq:dbus}
\end{align}

\shrink{
While we do not have a formal separation result, we believe that enforcing such
mandatory policies in the (untyped) $\pi$-calculus is difficult. On contrast, our
calculus permits to naturally model security mechanisms based on dedicated channels
that cannot be disclosed.  
}

\section{Related work}
\label{sec:discussion}

\ignore{ 
The idea that the restriction operator of the $\pi$-calculus ({\sf new}) and its scope extrusion mechanism
can model the possession and communication of a secret goes back to the spi
calculus~\cite{AbadiG99}. In this approach, one can devise specifications of protocols
that rely on  channel-based communication protected by {\sf new}, and establish the
desired security properties by using  equivalences representing  indistinguishablity even in the presence of a  Dolev-Yao attacker.
However, while the secrecy properties of some  protocols can be enforced  by relying on
dedicated channels, this is not  obvious when protocols describe
inter-site communication over the Transport/Internet layer. 
It is also not obvious  to implement scope extrusion over  dedicated channels. 
In the presence of open, untrusted networks the  secrecy of the protocol needs to be
recovered by means of other mechanisms, for instance by using  asymmetric
cryptography. However, preserving secrecy in the presence of  scope
extrusion is problematic~\cite{abadi98}, and eventually leads to complex cryptographic
protocols which rely on a set of certified authorities~\cite{popl07}. 
Moreover, the attacker can detect a communication, even if she cannot retrieve the content
of the message. 

Based on  these considerations, in this paper we argue that the
{\sf new} operator of the $\pi$-calculus does not adequately represent confidentiality.
We enrich the $\pi$-calculus with another operator for security, ({\em hide}), which
differs from {\sf new} in that it forbids the extrusion of a name and hence has a static
scope. To emphasize the
difference, we introduce a  spy process that represents a side-channel attack and
breaks some of the standard security equations for {\sf new}.
}

Many analysis and programming techniques for security have been developed for process
calculi. Among these, 
we would mention the security analysis enforced by means of static and dynamic
type-checking  (e.g.~\cite{CardelliGG05,Hennessy05,tgc05}), 
the verification of secure implementations and protocols that are 
protected by cryptographic encryption
(e.g.~\cite{BorealeNP01,AbadiFG02,AbadiBF07,popl07}),
and  programming models that consider a notion of location   
(e.g.~\cite{Hen07,SewellV03,CastagnaVN05}).

The paper~\cite{CardelliGG05} introduces a type system for a $\pi$-calculus with groups
that permits to control the distribution of resources: names can be received only by
processes in the scope of the group. The intent is, as in our paper, to preserve the
accidental or malicious leakage of secrets, even in the presence of untyped opponents. 
A limitation of~\cite{CardelliGG05} is that processes
that are not statically type-checked are interpreted as opponents trying to leak secrets. 
On contrast, our aim is to consider systems where processes  could dynamically join
the system at run-time; this permits us to analyze the secrecy of protocols composed by
trusted sub-systems that can grow in size of the number of the participants.
While devising an algorithm for type checking groups can be non-trivial 
(cf.~\cite{VasconcelosH93}), 
we note that actual systems do not often rely on types, even for local
communications. For instance D-Bus (cf. Section~\ref{sec:examples}) relies on a mandatory
access control policy  enforced at the kernel level through process IDs. Our
semantics-based approach appears as adequate to describe such low-level mechanisms.  

As discussed in the introduction, concrete implementations of $\pi$-calculi models
do protect communications by means of cryptography. The problem of devising a secure,
fully abstract  implementation has been first introduced
in~\cite{abadi98} and subsequently tackled for the join calculus in~\cite{AbadiFG02}. 
The paper~\cite{BorealeNP01} introduces a bisimulation-based technique to prove
equivalences of processes using cryptographic primitives; this can be used to show that
a  protocol does preserve secrecy. We follow a similar
approach and devise bisimulation semantics for establishing the secrecy of processes
running in an environment where the distribution of channels is controlled.  
The presence of a spy in our model is reminiscent of the  network abstraction
of~\cite{mik-mscs10}. In that paper, the network provides the low-level counter part of
the model where attacks based on bit-string representations, interception,  and
forward/reply
can be formalized. 

From the language design point of view, we share some similarity with the ideas
behind the boxed $\pi$-calculus~\cite{SewellV03}. 
A box in~\cite{SewellV03} acts as wrapper where we can confine untrusted process;
communication among the box and the context is subject to a fine-grained control  that
prevents the untrusted process to harm the protocol.
Our {\sf hide} operator is based on the symmetric principle:  processes
within the scope of an  {\sf hide} can run their protocol without be disturbed by the
context outside it.  
 

An interesting approach related to ours in spirit -- but not in conception or details --
is D-fusion~\cite{BorealeBM04}. 
The calculus has two forms of restriction:  A
"$\nu$" operator for name generation, and a "$\lambda$" operator that  behaves like an
existential quantifier and it can be seen as a generalization of an input binder.  Both
operators allow extrusion of the entities they declare but only the former guarantees
uniqueness.  In contrast our {\sf hide} operator is not meant as an existential nor as an
input-binder and  it prevents the extrusion of the name it declares.

\medskip\noindent{\bf Acknowledgements }
 We wholeheartedly thank the extremely competent, anonymous reviewers of EXPRESS 2012. 
 They went beyond the  call of duty in providing  excellent reports which have been very helpful to improve our paper. 

\bibliography{sessions}
\bibliographystyle{eptcs}

 
\end{document}